\documentclass[11pt]{article}
\usepackage{amssymb}
\usepackage{amsmath}
\usepackage{amsthm}

\usepackage{graphicx,color}
\usepackage{fullpage}
\usepackage{hyperref}
\def\openone{\leavevmode\hbox{\small1\kern-3.8pt\normalsize1}}

\newtheorem{theorem}{Theorem}
\newtheorem{lemma}[theorem]{Lemma}
\newtheorem{proposition}[theorem]{Proposition}

\newtheorem{conjecture}[theorem]{Conjecture}
\theoremstyle{definition}
\newtheorem{definition}[theorem]{Definition}

\newtheorem{remark}[theorem]{Remark}

\def\reff#1{(\ref{#1})}
\def\eps{\varepsilon}

\newcommand{\tr}{\mathop{\rm Tr}\nolimits}

\newcommand{\ket}[1]{|#1\rangle}

\newcommand{\cB}{{\cal B}}
\newcommand{\cC}{{\cal C}}
\newcommand{\cM}{{\cal M}}
\newcommand{\cD}{{\cal D}}
\newcommand{\cE}{{\cal E}}

\newcommand{\cN}{{\cal N}}
\newcommand{\cT}{{\cal T}}
\newcommand{\cX}{{\cal X}}
\newcommand{\cH}{{\cal H}}

\newcommand{\cP}{{\cal P}}

\newcommand{\cW}{{\cal W}}

\DeclareRobustCommand\openone{\leavevmode\hbox{\small1\normalsize\kern-.33em1}}
\newcommand{\identity}{I}%{\mathrm{\openone}}
\newcommand{\id}{\identity}
\newcommand{\be}{\begin{equation}}
\newcommand{\ee}{\end{equation}}
\newcommand{\bea}{\begin{eqnarray}}
\newcommand{\eea}{\end{eqnarray}}
\newcommand{\beas}{\begin{eqnarray*}}
\newcommand{\eeas}{\end{eqnarray*}}

\DeclareMathOperator*{\argmax}{\arg\max}

\newcommand{\MT}[1]{#1}

\numberwithin{equation}{section}

\begin{document}
\title{\bf On the Second-Order Asymptotics for\\Entanglement-Assisted Communication}
  
\author{Nilanjana Datta\thanks{Statistical Laboratory,
Centre for Mathematical Sciences, University of
Cambridge, Wilberforce Road, Cambridge CB3 0WB, UK}
\and Marco Tomamichel\thanks{Centre for Quantum Technologies, National
University of Singapore, Singapore 117543, Singapore, and
School of Physics, The University of Sydney, Sydney, Australia}
\and Mark M. Wilde\thanks{Hearne Institute for Theoretical Physics,
Department of Physics and Astronomy, Center for Computation and Technology,
Louisiana State University, Baton Rouge, Louisiana 70803, USA}}

\date{\today}

\maketitle

\begin{abstract}
The entanglement-assisted classical capacity of a quantum channel is known to provide the formal
quantum generalization of Shannon's classical channel
capacity theorem, in the sense that it admits a single-letter characterization
in terms of the quantum mutual information and does not increase in the presence
of a noiseless quantum feedback channel from receiver to sender. In this work, we investigate
second-order asymptotics of the entanglement-assisted classical communication task. That is, we
consider how quickly the rates of entanglement-assisted codes converge to the entanglement-assisted
classical capacity of a channel as a function of the number of channel uses and
the error tolerance. We define a quantum generalization of the mutual information
variance of a channel in the entanglement-assisted setting.
For covariant channels, we show that this quantity is equal to the channel dispersion, and thus completely characterize the convergence towards the entanglement-assisted classical capacity when the number of channel uses increases.
%More generally, we prove that the Gaussian approximation for a second-order coding rate is achievable for all quantum channels.} 
Our results also apply to entanglement-assisted quantum communication, due to the equivalence between entanglement-assisted classical and quantum communication established by the teleportation and super-dense coding protocols.
\end{abstract}

\section{Introduction}\label{intro}
Let us consider the transmission of classical information through a memoryless quantum channel. If the sender and receiver initially share entangled states which they may use in their communication protocol, then the information transmission is said to be {\em{entanglement-assisted}}. The entanglement-assisted classical capacity $C_{\text{ea}}(\mathcal{N})$
of a quantum channel $\mathcal{N}$
is defined to be the maximum rate at which a sender and receiver can communicate
classical information with vanishing error probability by using the channel $\mathcal{N}$
as many times as they wish and by using an arbitrary amount of shared
entanglement of an arbitrary form.
For a noiseless quantum channel, the entanglement-assisted classical capacity is twice as large 
as its unassisted one, an enhancement realized by the super-dense coding
protocol~\cite{PhysRevLett.69.2881}. This is in stark contrast to the setting of classical channels where additional shared randomness or entanglement does not increase the capacity.

\MT{Similarly,} for a noisy quantum channel, the presence of entanglement as an auxiliary resource can also lead to an enhancement of its classical
capacity \cite{PhysRevLett.83.3081,ieee2002bennett}. Somewhat surprisingly,
entanglement assistance is advantageous 
even for some entanglement-breaking channels \cite{horodecki03}, such as depolarizing channels with sufficiently high error probability. 
Bennett, Shor, Smolin and Thapliyal~\cite{ieee2002bennett} proved that the entanglement-assisted classical capacity $C_{\rm ea}(\cN)$ of a quantum channel $\cN$
is given by a remarkably simple, single-letter formula in terms of the quantum mutual information \MT{(defined in the following section)}. This is in contrast to the unassisted classical capacity of a quantum channel \cite{{H02},{PhysRevA.56.131}},
for which the best known general expression involves a regularization of the Holevo formula over infinitely many instances of the channel \cite{H09}. The regularization renders the explicit evaluation of the capacity for a general quantum channel intractable. 
The formula for the entanglement-assisted capacity is formally analogous to Shannon's well-known formula~\cite{bell1948shannon} for the capacity of a discrete memoryless classical channel, which is given in terms of the mutual information between the channel's input and output. The entanglement-assisted capacity does not increase under the presence of a noiseless quantum feedback channel from receiver to sender \cite{B04}, much like the capacity of a classical channel does not increase in the presence of a noiseless classical feedback link \cite{S56}. 

\MT{The formula for $C_{\rm ea}(\cN)$ derived in~\cite{PhysRevLett.83.3081}, however, is only relevant if the 
channel is available for as many uses as the sender and receiver wish, with there being no 
correlations in the noise acting on its successive inputs.\footnote{In other words, the channel is assumed to be {\em{memoryless}.}} 
To see this, let us consider the practical scenario in which a memoryless channel is used a 
finite number $n$ times, and let $\cN^n\equiv \cN^{\otimes n}$. Let $\log M_{\rm ea}^*(\cN^n,\eps)$ denote the maximum number of bits of information that can be transmitted through $n$ uses of the 
channel via an entanglement-assisted communication protocol, such that the average probability of failure is no larger than $\eps \in (0,1)$. Then~\cite{PhysRevLett.83.3081} and the strong converse~\cite{BDHSW12,BCR09} imply that}
\begin{align}
  \lim_{n \to \infty} \frac1{n} \log M_{\rm ea}^*(\cN^n,\eps) = C_{\rm ea}(\cN) .
\end{align}
The strong converse from \cite{GW13} implies that
\begin{align}
\log M_{\rm ea}^*(\cN^n,\eps) = n C_{\rm ea}(\cN) + O(\sqrt{n}), \label{eq:cap-expand}
\end{align}
for all $\varepsilon \in (0,1)$. The results of \cite{CMW14}
imply that this same expansion holds even when noiseless quantum feedback communication is allowed from receiver to sender.

\MT{We are interested in investigating the behavior of $M_{\rm ea}^*(\cN^n,\eps)$ for large but finite $n$ as a function of~$\eps$.
In this paper, we derive a lower bound on $\log M_{\rm ea}^*(\cN^n,\eps)$, for any fixed value of $\eps \in (0,1)$ and~$n$ large enough, of the following form:
\be\label{asymp_expansion}
\log M_{\rm ea}^*(\cN^n,\eps) \ge n C_{\rm ea}(\cN)  + \sqrt{n}\, b   + {{O}}(\log n).
\ee
The coefficient $b$ that we identify in this paper constitutes a
second-order coding rate.} The second-order coding rate obtained here depends on the channel as well as on the allowed error threshold~$\eps$, and we obtain an explicit expression for it in Theorem~\ref{th:main}. 
In addition, we conjecture that in fact
$\log M_{\rm ea}^*(\cN^n,\eps) = n C_{\rm ea}(\cN)  + \sqrt{n}\,b + \MT{{{o}}(\sqrt{n})}$ for all quantum channels. 
We show that this conjecture  is true for the class
of {\em{covariant channels}} \cite{H02}.

%We mention that prior work on the strong converse for entanglement-assisted capacity
%\cite{BDHSW12,BCR09,GW13}
%already identifies an upper bound of the following form: 
%\be
%\log M_{\rm ea}^*(\cN^n,\eps) \le n a + O(\sqrt{n}),
%\ee
%but the goal of a second-order analysis is to identify the constant for the $\sqrt{n}$
%term as a function
%of~$\eps$ and the channel.

Our lower bound on $\log M_{\rm ea}^*(\cN^n,\eps)$ resembles the asymptotic expansion for the maximum number of bits of information which can be transmitted through $n$ uses of a generic discrete, memoryless {\em{classical}} channel $\cW$, with an average probability of error no larger than $\eps$, denoted $\log M^*(\cW^n, \eps)$. The latter was first derived by Strassen in 1962~\cite{strassen62} and refined by Hayashi~\cite{Hay09} as well as Polyanskiy, Poor and Verd\'u~\cite{polyanskiy10}. It
is given by
\begin{equation}
\log M^*(\cW^n,\eps)= n C(\cW) + \sqrt{n V_\eps(\cW)}\Phi^{-1}(\eps)  + \MT{o(\sqrt{n})},   \label{eq:gauss}
\end{equation}
where $\cW^n$ denotes $n$ uses of the channel, $C(\cW)$ is its capacity given by Shannon's formula~\cite{bell1948shannon}, $\Phi^{-1}$ is the inverse of the cumulative distribution function of a standard normal random variable, and $V_\eps(\cW)$ is a property of the channel (which depends on $\eps$) called its $\eps$-{\em{dispersion}}~\cite{polyanskiy10}. The right hand side of \eqref{eq:gauss} is called the \emph{Gaussian approximation} of $\log M^*(\cW^n,\eps)$. %In fact, the Gaussian approximation  of the second-order asymptotics %(i.e., the occurrence of the term $\Phi^{-1}(\eps)$)  --- this is not really the point, it is rather the connection to the central limit theorem
%is a common feature of the second-order asymptotics for optimal rates of many other information-processing tasks such as data compression, communication, entanglement manipulation and randomness extraction (see, e.g.,~\cite{Hay09,polyanskiy10,P10,TH12,KH13,TT13,DL14} and references therein).
This result has recently been generalized to classical coding for quantum channels~\cite{TT13} and it was shown that a formula reminiscent of~\eqref{eq:gauss} also holds for the classical capacity of quantum channels with product inputs. In fact, the Gaussian approximation is a common feature of the second-order asymptotics for optimal rates of many other quantum information processing tasks such as data compression, communication, entanglement manipulation and randomness extraction (see, e.g.,~\cite{TH12,KH13,TT13,DL14} and references therein).

Even though we focus our presentation throughout on entanglement-assisted classical communication, we would like to point out that all of the results established in this paper apply to entanglement-assisted quantum communication as well. This is because the protocols of teleportation \cite{BBCJPW93} and super-dense coding \cite{PhysRevLett.69.2881} establish an equivalence between entanglement-assisted classical and quantum communication. This equivalence was noted in early work on entanglement-assisted communication \cite{PhysRevLett.83.3081}. That this equivalence applies at the level of individual codes
is a consequence of the development, e.g., in Appendix~B of \cite{LM14}, and as a result, the equivalence applies to second-order asymptotics as well. This point has also been noted in \cite{TB15}.

\MT{Finally, we note that a one-shot lower bound on $M_{\rm ea}^*(\cN,\eps)$ has already been derived in~\cite{DH13}.
% in terms of the min- and max-entropies~\cite{RennerThesis}. 
Moreover, in \cite{MW12} a one-shot upper bound was obtained.
% in terms of the hypothesis testing relative entropy. 
Even though these bounds converge in first order to the formula for the capacity obtained by Bennett {\em{et al.}}~\cite{ieee2002bennett}, 
neither of these works deals with characterizing second-order asymptotics.}

\MT{This paper is organized as follows. Section~\ref{sec_prelim}
 introduces the necessary notation and definitions. Section~\ref{sec_results} presents our main theorem and our conjecture. The proof of the theorem is given in Section~\ref{sec_proofs}. In Section~\ref{sec_proofs}, we also provide a proof of our 
conjecture 
for the case of covariant channels.
%, following the analysis by Matthews and Wehner in \cite{MW12}. 
We end with a discussion of
open questions in Section~\ref{sec_discussion}, summarizing the problems encountered when trying to prove the converse for general channels.}

\section{Notations and Definitions}\label{sec_prelim}

Let $\cB(\cH)$ denote the algebra of linear operators acting on a finite-dimensional Hilbert space $\cH$. Let ${\cP}({\cH})\subset \cB(\cH)$ be the set of positive semi-definite operators, and let ${\cD}({\cH})\subset \cP(\cH)$ denote the set of \emph{quantum states} (density matrices), ${\cD}({\cH}) :=\{\rho\in\cP(\cH): \tr\rho=1\}$.
%Throughout this paper, we restrict our considerations to finite-dimensional Hilbert spaces and 
We denote the dimension of a Hilbert space~${\cH}_A$ by $|A|$ and write $\cH_A \simeq \cH_{A'}$ when $\cH_A$ and $\cH_{A'}$ are isomorphic, i.e., if $|A| = |A'|$.
A quantum state $\psi$ is called pure if it is rank one; in this case, we associate with it an element $|\psi\rangle \in {\cH}$ such that $\psi = |\psi\rangle \langle\psi|$. The set of \emph{pure quantum states} is denoted $\cD_*(\cH)$.

For a bipartite operator $\omega_{AB} \in \cB(\cH_A \otimes \cH_B)$, let $\omega_{A} := \tr_B ( \omega_{AB})$ denote its restriction to the subsystem $A$, where $\tr_B$ denotes the partial trace over $B$. Let ${I}_{A}$ denote the identity operator on $\cH_A$, and let $\pi_A := {I}_A/|A|$ be the completely mixed state in ${\cal{D}}({\cal{H}}_A)$. 
 
 A \emph{positive operator-valued measure} (POVM) is a set $\{ \Lambda_{A}^x \}_{x \in \cX} \subset \cP(\cH_A)$ such that $\sum_{x \in \cX} \Lambda_A^x = \id_A$, where $\cX$ denotes any index set.
We use the convention that $\cE_{A\to B}$ refers to a \emph{completely positive trace-preserving} (CPTP) map $\cE_{A\to B}: \cB({\cH}_A) \to \cB({\cH}_B)$. We call such maps \emph{quantum channels} in the following.
%and $\cE_A$ refers to a CPTP map $\cE_A: \cB(\cH_A) \to \cB(\cH_A)$. 
The identity map on $\cB(\cH_A)$ is denoted~${\rm id}_A$.

We employ the cumulative distribution function for a standard normal random variable:
\be
 \Phi(a) := \frac{1}{\sqrt{2\pi}} \int_{-\infty}^a \rm{d} x \, \exp\left(-\frac{x^2}{2}\right).
\ee
and its inverse $\Phi^{-1}(\eps) := \sup \left\{ a \in \mathbb{R} \,|\, \Phi(a) \le \eps\right\}$.

\subsection{Entanglement-Assisted Codes}

We consider entanglement-assisted classical (EAC) communication through a noisy quantum channel, given by a CPTP map $\cN_{A \to B}$. The sender (Alice) and the receiver (Bob) initially share an arbitrary pure state $\ket{\varphi_{A'B'}}$, where without loss of generality we assume that $\cH_{A'} \simeq \cH_{B'}$, the system $A'$ being with Alice and the system $B'$ with Bob. The goal is to transmit classical messages from Alice to Bob, labelled by the elements of an index set $\cM$, through $\cN_{A \to B}$.

%Let us start by considering a single use of the channel. 
Without loss of generality, any EAC communication protocol can be assumed to have the following form: Alice encodes her classical messages into states of the system $A'$ in her possession. Let the encoding CPTP map corresponding to message $m \in \cM$ be denoted by $\cE^m_{A'\to A}$. Alice then sends the system $A$ through $\cN_{A \to B}$
to Bob. Subsequently, Bob performs a POVM $\{\Lambda^m_{BB'}\}_{m \in {\cal M}}$ on the system $BB'$ in his possession. This yields a classical register $\widehat{M}$ which contains his inference $\hat{m} \in \cM$ of the message sent by Alice. 

The above defines an \emph{EAC code} for the quantum channel
$\cN_{A \to B}$, which consists of a quadruple 
\begin{align}
\label{eq:code}
\cC = \Big\{\cM,\, |\varphi_{A'B'}\rangle,\,  \{\cE^m_{A'\to A}\}_{m\in\cM},\, \{\Lambda^m_{BB'}\}_{m\in\cM}  \Big\} . 
\end{align}
The size of a code is denoted as $|\cC| = |\cM|$.
The average probability of error for $\cC$ on $\cN_{A\to B}$ is
\be
p_{\text{err}}(\cN_{A \to B} ,\cC) := \Pr[M \neq \widehat{M}] = 1 - \frac{1}{|\cM|} \sum_m
\tr \Big(\Lambda^m_{BB'}\, \cN_{A \to B} \otimes {\rm id}_{B'} \big(\cE^m_{A'\to A} \otimes {\rm id}_{B'}(\varphi_{A'B'}) \big) \Big) .
\ee

The following quantity describes the maximum size of an EAC
code for transmitting classical information through a single use of $\cN_{A\to B}$ with average probability of error at most~$\eps$.

\begin{definition}
Let $\eps \in (0,1)$ and $\cN \equiv \cN_{A\to B}$ be a quantum channel. We define
\begin{equation}\label{pr}
M_{\rm ea}^*(\cN,\eps) :=  \max \big\{ m \in \mathbb{N} \,\big|\, \exists \, \cC : |\cC| = m \land p_{\text{err}}(\cN,\cC) \leq \eps \big\} ,
\end{equation}
where $\cC$ is a code as prescribed in~\eqref{eq:code}.
\end{definition}

We are particularly interested in the quantity $M_{\rm ea}^*(\cN^n, \eps)$,
where $n \in \mathbb{N}$ and $\cN^n \equiv \cN_{{A}^n\to B^n}^n =
\cN_{A_1 \to B_1} \otimes \ldots \otimes \cN_{A_n \to B_n}$ is the
$n$-fold memoryless repetition of $\cN$.

\subsection{Information Quantities}

For a pair of positive semi-definite operators $\rho$ and $\sigma$ with ${\rm supp}\, \rho \subseteq {\rm supp} \,\sigma$, 
the {\em{quantum relative entropy}} and the {\em{relative entropy variance}}~\cite{li12,TH12} are respectively defined as follows:\footnote{All logarithms in this paper are taken to base two.}
\begin{align}
D(\rho\|\sigma) &:= \tr \left[\rho\left(\log \rho - \log \sigma\right)\right], \qquad \textrm{and} \label{eq:rel-ent}\\
V(\rho\|\sigma) &:= \tr \left[\rho\left(\log \rho - \log \sigma - D(\rho\|\sigma)\right)^2\right].%
\end{align}
For a bipartite state $\rho_{AB}$, let us define the \emph{mutual information} $I(A:B)_{\rho} := D(\rho_{AB} \| \rho_A \otimes \rho_B)$. Similarly, we define the \emph{mutual information variance} $V(A:B)_{\rho} := V(\rho_{AB} \| \rho_A \otimes \rho_B)$. 

The EAC capacity of a quantum channel $\cN$ is defined as
\begin{align}
    C_{\rm ea}(\cN) := \lim_{\eps \to 0} \limsup_{n \to \infty} \frac{1}{n} \log M_{\rm ea}^*(\cN^n, \eps).
\end{align}

Bennett, Shor, Smolin and Thapliyal~\cite{ieee2002bennett} established that the EAC capacity for a quantum channel $\cN\equiv \cN_{{A}\to B}$ satisfies
 \begin{align}
    C_{\rm ea}(\cN) = \max_{\psi_{AA'}} I(A':B)_{\omega}, \qquad \textrm{where} \quad \omega_{A'B} = \cN_{A \to B} \otimes {\rm id}_{A'} (\psi_{AA'}), \label{eq:def-omega} 
  \end{align}
and the maximum is taken over all $\psi_{AA'} \in \cD_*(\cH_A \otimes \cH_{A'})$ with $\cH_{A'} \simeq \cH_A$. Its proof was later simplified by Holevo \cite{Hol01a}, and an alternative proof was given in \cite{HDW05}.

In analogy with~\cite{TT13}, the following definitions will be used to characterize our lower bounds on the second-order asymptotic behavior of $M_{\rm ea}^*(\cN^n,\eps)$.

\begin{definition}
  Let $\cN \equiv \cN_{A \to B}$ be a quantum channel.
 The set of \emph{capacity achieving resource states} on $\cN$ is defined as
  \begin{align}
    \Pi(\cN) := \argmax_{\psi_{AA'}} I(A':B)_{\omega} \subseteq \cD_*(\cH_A \otimes \cH_{A'}),
  \end{align}
  where $\omega_{A'B}$ is given in~\eqref{eq:def-omega}.
  The \emph{minimal mutual information variance} and the \emph{maximal mutual information variance} of $\cN$ are respectively defined as
  \begin{align}
    V_{\rm ea,\min}(\cN) := \min_{\psi_{AA'}} V(A' : B)_{\omega} \qquad \textrm{and} \qquad
    V_{\rm ea,\max}(\cN) := \max_{\psi_{AA'}} V(A' : B)_{\omega} ,
  \end{align}
  where the optimizations are over $\psi_{AA'} \in \Pi(\cN)$ and $\omega_{A'B}$ is given in~\eqref{eq:def-omega}.
\end{definition}

\section{Results}
\label{sec_results}
Our main result is stated in the following theorem,
which provides a second-order lower bound on the
maximum number of bits of classical message which can
be transmitted through $n$ independent uses of a
noisy channel via an entanglement-assisted protocol,
for any given allowed error threshold.

\begin{theorem}
  \label{th:main}
  Let $\eps \in (0,1)$ and let $\cN \equiv \cN_{A \to B}$ be a quantum channel. Then,
  \begin{align}
    \log M_{\rm ea}^*(\cN^n,\eps) \geq
    \begin{cases} n C_{\rm ea}(\cN) + \sqrt{n V_{\rm ea,\min}(\cN)}\ \Phi^{-1}(\eps) + K(n; \cN, \eps) & \text{if } \eps < \frac12 \\
    n C_{\rm ea}(\cN) + \sqrt{n V_{\rm ea,\max}(\cN)}\ \Phi^{-1}(\eps) + K(n; \cN, \eps) &  \text{else} \end{cases} \label{eq:direct}
  \end{align}
  where $K(n; \cN,\eps) = O(\log n)$.
\end{theorem}

\MT{The proof of Theorem~\ref{th:main} is split into two parts, Proposition~\ref{thm-one-shot} in Section~\ref{SEC_MAIN} and Proposition~\ref{pr:direct-second} in Section~\ref{SEC_SECOND}.
We first
derive a one-shot lower bound on $\log M_{\rm ea}^*(\cN,\eps)$ using 
a coding scheme that is a one-shot version of the coding scheme given in
\cite{HDW05} and reviewed in \cite[Sec.~20.4]{MW13}.
Our one-shot bound is expressed in terms of an entropic quantity called the hypothesis testing relative entropy~\cite{WR12}, which has its roots in early work on the quantum Stein's lemma \cite{HP91}  (see Section~\ref{sec_tech} for a definition). The relation between classical coding over a quantum channels and binary quantum hypothesis testing was first pointed out by Hayashi and Nagaoka~\cite{HN03}.

An asymptotic expansion for this quantity for product states was derived independently by Tomamichel and Hayashi~\cite{TH12} and Li \cite{li12}, and we make use of this to obtain our lower bound on $\log M_{\rm ea}^*(\cN^n,\eps)$ in the second step.}

\begin{remark}
In particular, Theorem~\ref{th:main} establishes that
  \begin{align}
    &\liminf_{n\to\infty} \frac{1}{\sqrt{n}} \Big( \log M_{\rm ea}^*(\cN^n,\eps) - n C_{\rm ea}(\cN) \Big) \geq \begin{cases} \sqrt{V_{\rm ea,\min}(\cN)}\ \Phi^{-1}(\eps) & \text{if } \eps < \frac12 \\
       \sqrt{V_{\rm ea,\max}(\cN)}\ \Phi^{-1}(\eps) & \text{else} \end{cases} .
  \end{align}
  In analogy with~\cite[Eq.~(221)]{polyanskiy10} and~\cite[Rm.~4]{TT13}, we define the EAC $\eps$-channel dispersion, for $\eps \in (0,1)\setminus \{\tfrac12\}$ as
  \begin{align}
    V_{\rm ea,\eps}(\cN) := \limsup_{n\to \infty} \frac{1}{n} \bigg( \frac{\log M_{\rm ea}^*(\cN^n,\eps) - n C_{\rm ea}(\cN)}{\Phi^{-1}(\eps)} \bigg)^2 .
  \end{align}
  Theorem~\ref{th:main} shows that $V_{\rm ea,\eps}(\cN) \leq V_{\rm ea,\min}(\cN)$ if $\eps < \frac12$ and $V_{\rm ea,\eps}(\cN) \geq V_{\rm ea,\max}(\cN)$ if $\eps > \frac12$.
\end{remark}

This leads us to the following conjecture:

\begin{conjecture}\label{conj:main}
  We conjecture that~\eqref{eq:direct} is an equality with $K(n; \cN,\eps) = o(\sqrt{n})$. In particular, we conjecture that the EAC $\eps$-channel dispersion satisfies 
  \begin{align}
    V_{\rm ea,\eps}(\cN) = \begin{cases} V_{\rm ea,\min}(\cN) & \text{if } \eps < \frac12 \\ V_{\rm ea,\max}(\cN) & \text{else} \end{cases}
  \end{align}
  and, thus, $\log M_{\rm ea}^*(\cN^n,\eps) = n C_{\rm ea}(\cN) + \sqrt{n V_{\rm ea,\eps}(\cN)}\ \Phi^{-1}(\eps) + o(\sqrt{n})$.
\end{conjecture}

We show that Conjecture~\ref{conj:main} is true for the class of covariant quantum channels. This follows essentially from
an analysis by Matthews and Wehner~\cite{MW12} which we recapitulate in Section~\ref{sec_converse} and the asymptotic expansion of the hypothesis testing relative entropy.

\section{Proofs}
\label{sec_proofs}

\subsection{Technical Preliminaries}
\label{sec_tech}

For given orthonormal bases $\{|i_{A}\rangle\}_{i=1}^d$ and $\{|i_{B}\rangle\}_{i=1}^d$ in isomorphic Hilbert spaces
${\cal{H}}_{A}\simeq{\cal{H}}_B\simeq{\cH}$ of dimension~$d$, we define a maximally entangled state of Schmidt rank $d$ to be
\begin{equation}\label{MES-m}
|\Phi_{AB}\rangle:= \frac{1}{\sqrt{d}} \sum_{i=1}^d |i_A\rangle\otimes |i_B\rangle.
\end{equation}
Note that if $d = 1$ then $|\Phi_{AB}\rangle$ is a product state.
We often make use of the following identity (``transpose trick''): for any operator~$M$,
\be\label{transpose}
(M_A \otimes I_B)|\Phi_{AB}\rangle = (I_A \otimes M_B^T)|\Phi_{AB}\rangle,
\ee
where $M_B^T := \sum_{i,j=1}^d |i\rangle_B \langle j|_A M_A |i\rangle_A \langle j|_B$ denotes the transpose.

\subsubsection{Distance Measures}

The \emph{trace distance} between two states $\rho$ and $\sigma$ is given by
\begin{align}
  \tfrac12 \Vert \rho - \sigma \Vert_1 = \max_{0 \leq Q \leq I} \tr \big( Q (\rho - \sigma) \big) = 
  \tr\bigl[\{\rho \ge \sigma\}(\rho-\sigma)\bigr] \label{eq:bd}
\end{align}
where $\{\rho\ge \sigma\}$ denotes the projector onto the subspace where the operator $\rho-\sigma$ is positive semi-definite.
The fidelity of $\rho$ and $\sigma$ is defined as
\begin{equation}\label{fidelity-aaa}
F(\rho, \sigma):=\left \Vert{\sqrt{\rho}\sqrt{\sigma}}\right \Vert_1.
\end{equation}
For a pair of pure states $\phi$ and $\psi$, the trace distance and fidelity satisfy the following relation:
\begin{equation}
\tfrac{1}{2} \left \Vert\phi - \psi\right \Vert_1 = \sqrt{1 - F^2(\phi, \psi))} .
\end{equation}

\subsubsection{Relative Entropies for One-Shot Analysis}

We will phrase our one-shot bounds in terms of the following relative entropy.
\begin{definition}
Let $\eps \in (0,1)$, $\rho\in {\cal D}({\cal H})$ and $\sigma \in  {\cal P}({\cal H})$. Then, 
the \emph{hypothesis testing relative entropy}~\cite{WR12} is defined as %follows:
\be
D_H^\eps(\rho\|\sigma) := -\log \beta_\eps (\rho\|\sigma),
\ee
where 
\be\label{beta}
\beta_\eps (\rho\|\sigma) := \min \big\{ \tr (Q\sigma) : 0\le Q \le I  \wedge \tr(Q\rho) \ge 1-\eps \big\}.
\ee
\end{definition}

Note that when $\sigma \in  {\cal D}({\cal H})$, $\beta_\eps (\rho\|\sigma)$
has an interpretation as the smallest type-II error
of a hypothesis test between $\rho$ and $\sigma$, when the type-I error is at most~$\eps$.
The following lemma lists some properties of $D_H^\eps(\rho\|\sigma)$. 

\begin{lemma}\label{props-hypo}
Let $\eps \in (0,1)$. 
The hypothesis testing relative entropy has the following properties:
\begin{enumerate}
\item For any $\rho \in \cD(\cH)$, $\sigma^\prime \ge \sigma \ge 0$ we have $D_H^\eps(\rho\|\sigma)\ge D_H^\eps(\rho\|\sigma^\prime)$.
\item For any $\rho \in \cD(\cH)$, $\sigma \geq 0$, $\alpha >0$, we have $D_H^\eps(\rho\|\alpha \sigma)= D_H^\eps(\rho\|\sigma) - \log \alpha$.
\item For any classical-quantum state
\be
\rho_{XB} = \sum_{x \in {\cal X}} p(x) |x\rangle \langle x| \otimes \rho^x_B \in \cD(\cH_X \otimes \cH_B),
\ee
for any 
$\sigma_X= \sum_{x \in {\cal X}} q(x) |x\rangle \langle x|$ (with $\{p(x)\}_{x \in {\cal X}}$ and
 $\{q(x)\}_{x \in {\cal X}}$ probability distributions on ${\cal X}$), and for any $\sigma_B \in \cD(\cH_B)$, we have
\be
  D_H^\eps(\rho_{XB}\|\sigma_X \otimes \sigma_B) \ge \min_{x \in {\cal X}} D_H^\eps(\rho_{B}^x\|\sigma_B),
\ee
\item For any $\delta \in (0,1-\eps)$, $\rho, \rho' \in \cD(\cH)$ with $\frac12 \left\Vert\rho - \rho^\prime\right\Vert_1 \le \delta$, and $\sigma \in \cP(\cH)$, we have $D_H^\eps(\rho'\|\sigma) \le D_H^{\eps+\delta}(\rho\| \sigma)$.
\end{enumerate}
\end{lemma}

Properties 1--3 can be verified by close inspection and we omit their proofs. %(see also~\cite{TT13}).

\begin{proof}[Proof of Property 4]
Consider $Q$ to be the operator achieving the minimum in the definition of
$ \beta_\eps(\rho'\|\sigma)$, i.e.\
\begin{align}
% D_H^\eps(\rho'\|\sigma) = \max_{0\le Q' \le I\atop{ \tr (Q'\rho') \ge 1- \eps}}\left[- \log \tr (Q' \sigma)\right]
%  =  - \log \tr (Q \sigma).
 D_H^\eps(\rho'\|\sigma) =  - \log \tr (Q \sigma) \qquad \textrm{and} \qquad \tr( Q \rho') \geq 1-\eps.
\end{align}
From~\eqref{eq:bd}, we have
$\tr\left[Q(\rho' -\rho)\right] \le \frac12 \left\Vert\rho - \rho'\right\Vert_1 \le \delta.$
Hence, $\tr (Q\rho) \ge \tr (Q\rho') - \frac12 \left\Vert\rho - \rho^\prime\right\Vert_1 \ge 1- \eps - \delta$, and
\begin{align}
 D_H^\eps(\rho'\|\sigma) & \le \max_{0\le Q' \le I\atop{ \tr (Q'\rho) \ge 1- \eps -\delta}}\left[- \log \tr (Q' \sigma)\right]  = D_H^{\eps+\delta}(\rho\| \sigma) ,
\end{align}
which concludes the proof.
\end{proof}
The following result, established independently in~\cite[Eq.~(34)]{TH12} and  \cite{li12}, plays a central role in our analysis.
\begin{lemma}[\cite{TH12,li12}]\label{second-order}
Let $\eps \in (0,1)$ and let $\rho, \sigma \in \cD(\cH)$. Then,
\be
 D_H^\eps(\rho^{\otimes n} \| \sigma^{\otimes n}) = nD(\rho\|\sigma) + \sqrt{n V(\rho\|\sigma)} \Phi^{-1}(\eps) + O(\log n) .
\ee
\end{lemma}

Two other generalized relative entropies which are relevant for our analysis
are the {\em{collision relative entropy}} and the {\em{information-spectrum relative entropy}}~\cite[Def.~8]{TH12}. For any pair of positive semi-definite operators $\rho$ and $\sigma$ satisfying the condition ${\rm supp} \,\rho \subseteq {\rm supp} \,\sigma$, they are respectively defined as follows:
\be\label{collision}
D_2(\rho\|\sigma) := \log \left( \tr \left( \sigma^{-1/4} \rho \sigma^{-1/4}\right)^2 \right) ,
\ee
and, for any $\eps \in (0, 1)$,
\be\label{infospec}
D_s^\eps(\rho\|\sigma) := \sup\left\{R \,|\, \tr\left(\rho \big\{\rho \le 2^R \sigma\big\} \right) \le \eps\right\}  ,
\ee
where we write $A \geq B$ if $A - B$ is positive semidefinite. 
The following result, proved in~\cite[Thm.~4]{BG13}, relates these quantities.
\begin{lemma}[\cite{BG13}]
Let $\eps, \lambda \in (0,1)$ and $\rho, \sigma \in \cD(\cH)$. Then,
\be
  2^{D_2\left( \rho \| \lambda \rho +(1-\lambda)\sigma \right)} \ge (1-\eps) 
\left[\lambda + (1-\lambda) 2^{-D_s^\eps(\rho\|\sigma)} \right]^{-1} .
\ee
\label{lemBG}
\end{lemma}
Finally, the following lemma provides a useful relation between the hypothesis testing relative entropy and the information spectrum relative entropy~\cite[Lm.~12]{TH12}.
\begin{lemma}[\cite{TH12}]
Let $\eps \in (0,1)$, $\delta \in (0, 1-\eps)$, $\rho \in \cD(\cH)$, and $\sigma \in {\cP}(\cH)$. Then,
$D_H^{\eps}(\rho\|\sigma) \ge D_s^{\eps} (\rho\|\sigma) \ge D_H^{\eps+\delta}(\rho\|\sigma) + \log \delta.$
\label{info-hypo}
\end{lemma}

\subsection{One-Shot Achievability}
\label{SEC_MAIN}%

Our protocol is modeled after~\cite{ieee2002bennett}.
Consider a quantum channel $\cN_{A \to B}$ and introduce an auxiliary Hilbert space $\cH_{A'} \simeq \cH_A$.
%For our one-shot achievability bound, we consider resource states $\vartheta_{AA'}$ of the following %form:
%
Let
\be\label{directsum}
\cH_{A} \otimes \cH_{A'} = \bigoplus_t \cH_{A}^t \otimes \cH_{A'}^t, \qquad \ \cH_{A}^t \simeq \cH_{A'}^t
\ee
be a decomposition of $\cH_{A} \otimes \cH_{A'}$, and set $d_t = |\cH_{A}^t|$.
We assume that $|\vartheta_{AA'}\rangle$ can be written as a superposition of maximally entangled states:
\be\label{decomp} |\vartheta_{AA'}\rangle = \sum_t \sqrt{p(t)}\, |\Phi^t\rangle,
\ee
where $|\Phi^t\rangle$ denotes a maximally entangled state of Schmidt rank $d_t$ in $\cH_{A}^t \otimes \cH_{A'}^t$ and $p(t)$ is some probability distribution so that
$\sum_t p(t) = 1$. %We take the convention that $|\Phi^t\rangle$ is a product state if $d_t=1$.
 % {\em{Note that any bipartite pure state can be expressed in the following form. In particular, in the simplest case in which the Hilbert spaces ${\cal H}_t$ are one-dimensional, \reff{Schmidt} just corresponds to the Schmidt decomposition of $|\Phi(t)\rangle$ with Schmidt coefficients $p(t}/d(t)$.}} 
    
\begin{proposition}\label{thm-one-shot}
Let $\eps \in (0,1)$, $\delta \in (0, \frac{\eps}2)$ and let $\cN \equiv \cN_{A\to B}$ be a quantum channel. Then for any 
%prior shared entangled state 
$\vartheta_{AA'}$ of the form~\eqref{decomp}, we have
\begin{align}
\log M_{\rm ea}^*(\cN,\eps) \ge  D_H^{\eps - 2\delta}(\cN_{A \to B}( \vartheta_{AA'}) \,\|\, \cN_{A \to B}(\kappa_{AA'} ) ) - f(\eps, \delta),
\end{align}
where $f(\eps, \delta) := \log \frac{1-\eps}{\delta^2}$, $\kappa_{AA'} := \sum_t p(t)\, \pi_{A}^t \otimes \pi_{A'}^t$, and $\pi_A^t$ is the maximally mixed state on $\cH_{A}^t$.
\end{proposition}

\begin{remark}
Note that the hypothesis testing relative entropy on the right hand side is not reminiscent of a mutual information type quantity since the second argument is not a product state. 
% It is worth pointing out that in spite of an aesthetically pleasing bound derived for classical-quantum channels~\cite[Thm.~2]{WR12}, there is in general no reason to expect meaningful one-shot bounds in terms of a mutual information type quantity for coding problems.
\end{remark}

\begin{proof}
Consider the set 
\be
{\cal S}:= \big\{\left((x_t, z_t, b_t)\right)_t \,\big|\, x_t, z_t \in \{0,1,\cdots, d_t-1\}, b_t \in \{0,1\}\big\},
\ee
where the index $t$ labels the Hilbert spaces of the decomposition in~\reff{directsum}. For any $s \in {\cal S}$, consider the following unitary operator in $\cB(\cH_{A})$:
\be\label{unitaryop}
U_{A}(s) := \bigoplus_t (-1)^{b_t} X(x_t) Z(z_t),
\ee
where $X(x_t)$ and $Z(z_t)$ are the Heisenberg-Weyl operators defined in~Appendix~\ref{app:hw}. 

For any $M \in \mathbb{N}$, we now construct a random code as follows. Let $\cM = \{1, 2, \ldots, M \}$. We set $A' \equiv A$ (i.e., we use the labels interchangeably), and $\cH_{B'} \simeq \cH_A$. We consider the resource state $\varphi_{AB'} = {\rm id}_{A' \to B'}(\vartheta_{AA'})$.
For each message $m\in \cM$, choose a \emph{codeword}, $s_m$, uniformly at random from the set ${\cal S}$. The encoding operation, $\{\mathcal{E}_{A}^m \}_{m \in \cM}$, is then given by the (random) unitary $U(s_m)$ as prescribed above. In particular,
\begin{align}
  \cE^m_{A} \otimes {\rm id}_{B'}(\varphi_{AB'}) = \phi^{s_m}_{AB'}, \quad \textrm{where} \quad
  |\phi^{s_m}_{AB'}\rangle := \left(U_{A}(s_m) \otimes \id_{B'}\right)|\varphi_{AB'}\rangle.
   \label{eq:enc}
\end{align}
We denote the corresponding channel output state by
$\rho^{s_m}_{BB'}:=\cN_{A \to B}\left(\phi^{s_m}_{AB'}\right)$ and use ``pretty good'' measurements for decoding. 
These are given by the POVM $\{\Lambda_{BB'}^m\}_{m \in \cM}$, where
\begin{align}
 \Lambda_{BB'}^m :=\left( \sum_{m'\in {\cal M}}\rho^{s_{m'}}_{BB'}\right)^{-\frac{1}{2}}\rho^{s_m}_{BB'}\left(\sum_{m'\in {\cal M}}\rho^{s_{m'}}_{BB'}\right)^{-\frac{1}{2}} . \label{eq:dec}
\end{align}

Let us now analyze the code $\cC = \{\cM, \varphi_{AB'},  \{\cE^m_{A'}\}_{m\in\cM}, \{\Lambda^m_{BB'}\}_{m\in\cM}  \}$ given by~\eqref{eq:enc} and~\eqref{eq:dec}, where we recall that $s_m$ is a random variable.
For this purpose, consider the random state
\begin{align}
\sigma_{MSBB'} := \frac{1}{M} \sum_{m \in {\cal M}} |m \rangle \langle m|_M \otimes |s_m\rangle \langle s_m|_S \otimes \rho^{s_m}_{BB'}.
\end{align}
Then, following Beigi and Gohari~\cite[Thm.~5]{BG13}, we find that the average probability of successfully inferring the sent message can be expressed as
\begin{align}\label{succ}
p_{{\rm succ}}(\cC, \cN) &:= 1 - p_{\text{err}}(\cC, \cN) =
\frac{1}{M} \sum_{m \in {\cal M}} \tr (\Lambda_{BB'}^m  \rho^{s_m}_{BB'}) \\
&= \frac{1}{M} 2^{D_2\left(\sigma_{MSBB'}\| \sigma_{MS} \,\otimes\, \sigma_{BB'} \right)}.
\end{align}  
Moreover employing both the data-processing inequality and joint convexity
of the collision relative entropy as in~\cite{BG13}, we establish the following lower
bound on the expected value of ${p}_{{\rm succ}}$ with respect
to the randomly chosen codewords:
\be
\mathbb{E}\left(p_{{\rm succ}}(\cC, \cN)\right)\ge \frac{1}{M} 2^{D_2\left( \mathbb{E}(\sigma_{SBB'})\|  \mathbb{E}(\sigma_{S} \,\otimes\, \sigma_{BB}) \right)}. %
\ee
Note that 
% $$\omega_{SBB'}  = \frac{1}{M} \sum_{s \in {\cal S}} |s\rangle \langle s| \otimes 
% (\cN\otimes {\rm id}_{B'}) 
% \left((U_{A'}(s)\otimes I_{B'})\varphi_{A'B'} ( U_{A'}^\dagger(s)\otimes I_{B'})\right)$$ 
\begin{align}
 \mathbb{E}(\sigma_{S} \otimes \sigma_{BB'}) &=  \mathbb{E}\left(\frac{1}{M^2} \sum_{m \in {\cal M}} |s_m\rangle \langle s_m| \otimes \rho^{s_m}_{BB'} \right) +  \mathbb{E}\left(\frac{1}{M^2} \sum_{m, m' \in {\cal M}\atop{m' \ne m}} |s_m\rangle \langle s_m| \otimes \rho^{s_{m'}}_{BB'} \right) \\
&= \frac{1}{M} \rho_{SBB'} + \left( 1- \frac{1}{M}\right) \rho_S \otimes \rho_{BB'}, %
\end{align}
where 
\begin{align}
\rho_{SBB'} &:= \mathbb{E}(\sigma_{SBB'}) = \cN_{A \to B} \left(\frac{1}{|{\cal S}|} \sum_{s \in {\cal S}} |s\rangle \langle s| \otimes U_{A}(s) \varphi_{AB'} U_{A}^\dagger(s)\right)\\
&= \frac{1}{|{\cal S}|} \sum_{s \in {\cal S}} |s\rangle \langle s| \otimes \cN_{A \to B} \left(U_{A}(s) \varphi_{AB'} U_{A}^\dagger(s)\right),\label{c-q}
\end{align}
and $\rho_S$ and $\rho_{BB'}$ are the corresponding reduced states on the systems $S$ and $BB'$, respectively. In particular, defining $V(x_t, z_t) := X(x_t)Z(z_t)$, using the decomposition \reff{decomp} of the state $\ket{\varphi_{AB'}}$ and the definition \reff{unitaryop} of the unitary operators $U_{A}(s)$, we find that
\begin{align}
\rho_{BB'}&=  \cN_{A \to B}
\left(\frac{1}{|{\cal S}|}\sum_{s \in {\cal S}} U_{A}(s)  \left(\sum_{t,t'} \sqrt{p(t)p(t')} |\Phi^t\rangle \langle \Phi^{t'}|\right)U_{A}^\dagger(s)\right)\\
&=  \cN_{A \to B} \left(\sum_{t}{p(t)} \frac{1}{d_t^2} \sum_{x_t, z_t =0}^{d_t - 1} V(x_t, z_t)   |\Phi^t\rangle \langle \Phi^t|V^\dagger(x_t, z_t)\right) \\
& +  \cN_{A \to B} \left(\sum_{t,t'\atop{t'\ne t}} \sqrt{p(t)p(t')}\ \frac{1}{4} \sum_{b_t, b_{t'} \in \{0,1\}}(-1)^{b_t+ b_{t'}} \frac{1}{d_t^2d_{t'}^2} \sum_{x_t, z_t =0}^{d_t - 1} \sum_{x_{t'}, z_{t'} =0}^{d_{t'} - 1}   V(x_{t},z_{t}) |\Phi^t\rangle \langle \Phi^{t'}| V^\dagger(x_{t'},z_{t'})\right) %
\end{align}
can be written as the sum of a diagonal ($t = t'$) and an off-diagonal ($t \neq t'$) term.
It can be verified (see, e.g.,~\cite[pp.~504--505]{MW13}) that the off-diagonal term vanishes and in fact
\be\label{decouple}
\rho_{BB'} = \sum_t p(t) \, \cN_{A \to B}(\pi_{A}^t) \otimes \pi_{B'}^t ,
\ee
where $\pi_{A}^t = \tr_{B'}(\Phi^t)$ and $\pi_{B'}^t = \tr_{A}(\Phi^t)$ are completely
mixed states. The above identity follows from the fact that applying a Heisenberg-Weyl operator uniformly at random completely randomizes a quantum state, yielding a completely mixed state.

Hence, for any $0 < \delta < \eps$, we have
\begin{align}\label{last} 
\mathbb{E}\left(p_{{\rm succ}}(\cC, \cN)\right)& \ge \frac{1}{M} 2^{D_2\left(\rho_{SBB'}\| \frac{1}{M} \rho_{SBB'} + (1- \frac{1}{M})(\rho_{S} \otimes \rho_{BB'} \right)}\\
& \ge \frac{1-(\eps-\delta)}{1+(M-1)\,2^{-D_s^{\eps-\delta}(\rho_{SBB'}\| \rho_S \otimes \rho_{BB'})}},
\end{align}
where the last line follows from Lemma~\ref{lemBG}. 
Thus, provided that
\begin{align}
M &\leq  %\left\lceil 
  \frac{\delta}{1-\eps} 
2^{D_s^{\eps - \delta}\left(\rho_{SBB'}\| \rho_S \otimes \rho_{BB'}\right)} + 1
%\right\rceil,
\end{align}
the random code satisfies $\mathbb{E}\left(p_{{\rm succ}}(\cC, \cN)\right) \ge 1- \eps$.
In particular, there exists a (deterministic) code which satisfies $p_{{\rm succ}}(\cC, \cN) \ge 1- \eps$.
Hence, we conclude that 
\begin{align}
\log M_{\rm ea}^*(\cN,\eps) & \ge D_s^{\eps - \delta}(\rho_{SBB'}\| \rho_S \otimes \rho_{BB'}) +  \log \frac{\delta}{1-\eps}\\
& \ge
%D_H^{\eps - 2\delta}\left(\rho_{SBB'}\| \rho_S \otimes \rho_{BB'}\right)
%+ \log \frac{\delta^2}{1-\eps}  \\
%& =
D_H^{\eps - 2\delta}(\rho_{SBB'}\| \rho_S \otimes \rho_{BB'}) - f(\eps, \delta),
\label{stp1}
\end{align}
where we require that $\eps > 2 \delta$ and use
\be\label{fed}
f(\eps, \delta) = \log \frac{1-\eps}{\delta^2}. %
\ee
The inequality in \reff{stp1} follows from Lemma~\ref{info-hypo}.
Further, since $\rho_{SBB'}$ is a classical-quantum state as seen in \reff{c-q},
by item $3$ of Lemma~\ref{props-hypo} we have
\be
 D_H^{\eps - 2\delta}(\rho_{SBB'}\| \rho_S \otimes \rho_{BB'})
\ge \min_{s \in {\cal S}}  D_H^{\eps - 2\delta}(\rho_{BB'}^s\| \rho_{BB'}), %
\ee
where 
\be
\rho_{BB'}^s=\cN_{A \to B} \left(U_{A}(s) \varphi_{AB'} U_{A}^\dagger(s)\right). %
\ee
Using the decomposition \reff{decomp} of the state $\ket{\varphi_{AB'}}$ and the transpose trick \reff{transpose} we can write
\be\label{eqi}
\rho_{BB'}^s= U_{B'}^T(s)\cN_{A \to B}\left( \varphi_{AB'}\right) U_{B'}^{T\dagger}(s).
\ee
Further, from \reff{decouple} it follows that 
\be U_{B'}^T(s) \rho_{BB'} U_{B'}^{T\dagger}(s) = \rho_{BB'}.
\label{eqii}
\ee
Hence, \reff{eqi}, \reff{eqii}, \reff{decouple}, and the invariance of the hypothesis testing relative entropy under the same unitary on both states imply that 
\begin{align}
D_H^{\eps - 2\delta}(\rho_{BB'}^s\| \rho_{BB'})=
D_H^{\eps - 2\delta}\Bigg(\cN_{A \to B}\left( \varphi_{AB'}\right) \Bigg\| \sum_t p(t)\left( \cN_{A \to B}(\pi_{A}^t)\right) \otimes \pi_{B'}^t\Bigg),
\label{eqiii}
\end{align}
From \reff{stp1} and \reff{eqiii} we obtain the statement of the proposition. 
\end{proof}

\begin{remark}
Alternatively, one may also employ the one-shot achievability result of Hayashi and Nagaoka~\cite{HN03}
(in the form of~\cite{WR12}), which
leads to
the following bound on the one-shot $\eps$-error entanglement-assisted capacity.
Let $\eps \in (0,1)$. Then, for any $\delta \in (0,\eps)$ and 
for any %prior shared entangled state 
$\ket{\vartheta_{AA'}}$ with
decomposition \reff{decomp}, 
%the number of bits of classical message ($ \log |\cM|$) which can be transmitted
%over a single use of a noisy quantum channel ${\cal N}_{A^\prime \to B}$,
%with an average probability of error at most $\eps$, satisfies the bound
we have
\be 
\log M_{\rm ea}^*(\cN,\eps) \ge  D_H^{\eps - \delta}\Bigg(\cN_{A\to B}\left( \vartheta_{AA'}\right) \,\Bigg\|\, \sum_t p(t)\left( \cN_{A \to B}(\pi_{A}^t)\right) \otimes \pi_{A'}^t\Bigg) - 
\log\frac{4\eps}{\delta^2}.
\ee
The proof of this lower bound uses the same coding scheme as
given above while employing the error analysis and decoder given in \cite{HN03}.
\end{remark}

\subsection{Second-Order Analysis for Achievability}
\label{SEC_SECOND}

Theorem~\ref{th:main} is a direct corollary of the following result, for an appropriate choice of $\psi_{AA'}$.

\begin{proposition}
  \label{pr:direct-second}
  Let $\eps \in (0,1)$, $\cN \equiv \cN_{A\to B}$ be a quantum channel, and $\psi_{AA'} \in \cD_*(\cH_A \otimes \cH_A')$, where $\cH_{A'} \simeq \cH_A$. Then, we have
  \begin{align}
    \log M_{\rm ea}^*(\cN^n,\eps) \geq n I(A':B)_{\omega} + \sqrt{n V(A':B)_{\omega}}\,\Phi^{-1}(\eps) + K(n; \cN,\eps,\psi_{AA'}) ,
  \end{align}
  where $\omega_{A'B} = \cN_{A \to B} \otimes {\rm id}_{A'} (\psi_{AA'})$ and $K(n; \cN,\eps,\psi_{AA'}) = O(\log n)$.
\end{proposition}

\begin{proof}
We intend to apply Proposition~\ref{thm-one-shot} to the channel $\cN^n:=\cN^{\otimes n}$ for a fixed $n$. For this purpose, let us
first construct an appropriate resource state $\vartheta_{A^nA'^n}$.
We write
\be
 |\psi_{AA'}\rangle = \sum_{x \in {\cal X}} \sqrt{q(x)}\, |x\rangle_A \otimes |x\rangle_{A'}
\ee
in its  Schmidt decomposition,
where ${\cal X} = \{ 1,2, \cdots, d\}$ with $d = |\cH_{A}| = |\cH_{A'}|$ and define $\rho_{A'} = \tr_{A} (\psi_{AA'})$.
For a sequence $x^n = (x_1, x_2, \ldots, x_n) \in \cX^n$, we write $|x^n\rangle_{A^n} = \ket{x_{1}}_{A_1} \otimes \ket{x_{2}}_{A_2} \otimes \cdots \otimes \ket{x_{n}}_{A_n}$.
We denote type classes (for sequences of length $n$) by $\cT^t$, i.e.\
$\cT^t = \{x^n \in {\cal X}^n \,:\, P_{x^n} = t\}$
where $P_{x^n}$ denotes the empirical distribution of the sequence $x^n \in {\cal X}^n$. The set of empirical distributions is denoted $\cP_n$. (We refer to Appendix~\ref{sc:types} for a short overview of the method of types and relevant results.) 
We consider the decomposition 
\be 
\left(  \cH_{A} \otimes \cH_{A'}\right)^{\otimes n} =  \bigoplus_{t \in \cP_n} {\cH_{A^n}^t} \otimes \cH_{A'^n}^t,
\ee
where $\cH_{A^n}^t = \textrm{span} \big\{ |x^n\rangle_{A^n} \,\big|\, x^n \in \cT^t \big\}$
as in~\eqref{directsum}. Notably, since $\psi_{AA'}^{\otimes n}$ is a tensor-power state, we can write
\be
|\psi_{AA'}\rangle^{\otimes n} = \sum_{t \in \cP_n} \sqrt{p'(t)}\, |\Phi^t\rangle,
\ee
where $|\Phi^t\rangle \in \cH_{A^n}^t \otimes \cH_{A'^n}^t$ denotes a maximally entangled state of Schmidt rank $d_t = |\cT^t|$, and
\be
p'(t) :=  \sum_{x^n \in \cT^t} q^n(x^n), \quad \textrm{where} \quad q^n(x^n)=\prod_{i=1}^n q(x_i) \,.
\ee
%
%Moreover, let $\rho_{A'} = \tr_{A} \psi_{AA'} = \sum_{x \in {\cal X}} q(x)\, |x\rangle\langle x|_{A'}$. Then
%\be\label{rhon}
%\rho_{A'}^{\otimes n} = \sum_{x^n \in {\cal X}^n} q^n(x^n) |x^n\rangle\langle x^n|_{A'^n} .
%\ee

Now, fix a small $\mu > 0$ and consider a restriction of $|\psi_{AA'}\rangle^{\otimes n}$ to types $\mu$-close to $q$. More precisely, we consider the set $\cP_n^{q,\mu} := \{ t \in \cP_n \,|\, D(t\|q) \leq \mu \}$ and define
\begin{align}
\label{aux}
&|\vartheta_{A^{n}A'^{n}}\rangle := \sum_{t \in \cP_n^{q,\mu}} \sqrt{p(t)}\, |\Phi^t\rangle, \quad \textrm{where} \quad p(t) = \frac{p'(t)}{\alpha}, \quad \textrm{and} \\
&\alpha :=  \sum_{t \in \cP_n^{q,\mu}} p'(t) 
    \ =\!\! \sum_{x^n \in \cX^n \atop{ D(P_{x^n}\|q) \le \mu }}q^n(x^n) \ge 1 -  2^{-n\left( \mu - |{\cal X}| \frac{\log (n+1)}{n}\right)},
\label{alpha-bd}
\end{align}
where the last inequality follows from \reff{type6} in Appendix~\ref{sc:types}.
%\textcolor{red}{WE SEEM TO BE GOING BACK AND FORTH BETWEEN $\varphi$ and $\vartheta$ IN THIS SECTION. IS THIS INTENDED?} \textcolor{blue}{No, it should be $\vartheta$ all the time here, $\varphi$ only appears in the code definition and the proof of the one-shot proposition. I changed it.}
%
Note that 
\begin{align}
\frac{1}{2} \left\Vert\vartheta_{A^{\prime n}B^{\prime n}} - \psi_{A'B'}^{\otimes n}\right\Vert_1
&= \sqrt{1 - F^2\!\left(\vartheta_{A^{\prime n}B^{\prime n}}, \psi_{A'B'}^{\otimes n} \right)}
= \sqrt{ 1 -  \alpha} \\
& \le  2^{-\frac{n}{2} \left( \mu - |{\cal X}|\frac{ \log (n+1)}{n}\right)} =:  g(n,\mu),
\label{gnmu}
\end{align}
where the last inequality follows from \reff{alpha-bd}.

Next, recall that ${\cN}^n \equiv \left({\cal N}_{A \to B}\right)^{\otimes n}$. Then by Proposition~\ref{thm-one-shot}, for fixed $\eps>0$ and $0< 2\delta < \eps$ and $\vartheta_{A^{n}A'^{n}}$ given in~\reff{aux}, we establish that
\be 
\log M_{\rm ea}^*(\cN^n,\eps) \ge  D_H^{\eps - 2 \delta}\bigg(\cN^n\left( \vartheta_{A^{n}A'^{n}}\right) \, \bigg\| \, \sum_{t\in \cP_n^{q,\mu}} p(t)\left( \cN^n(\pi_{A^n}^t)\bigg) \otimes \pi_{A'^{n}}^t \right) - f(\eps, \delta),
\label{bd1shot}
\ee 
where $f(\eps, \delta)$ is given by \reff{fed}, and $\pi_{A^{n}}^t$ and $\pi_{A'^{n}}^t$ are completely mixed states. In particular, for any $t$ with $D(t\|q) \leq \mu$, we have
\begin{align}
\pi_{A'^{n}}^t =  \frac{1}{d_t}\sum_{x^n \in 
\cT^t}|x^n \rangle \langle x^n| &\le (n+1)^{|{\cal X}|}  2^{n\mu} \sum_{x^n \in T_t }  q^n({x^n}) |x^n \rangle \langle x^n|  \\
& \le \underbrace{(n+1)^{|{\cal X}|}  2^{n\mu}}_{=:\, \gamma_{n,\mu}} \sum_{x^n \in {\cal X}^n}  q^n({x^n}) |x^n \rangle \langle x^n| = \gamma_{n,\mu}\, \rho_{A'}^{\otimes n}\, . \label{bd2}
\end{align}
The first inequality in \reff{bd2} follows from \reff{type4} in Appendix~\ref{sc:types}, which is a consequence of the fact that $D(t\|q) \le \mu$.

Next, we use \reff{bd1shot} and \reff{bd2} to obtain
% and items $1$, $2$ and $4$ of , 
\begin{align} 
& \log M_{\rm ea}^*(\cN^n,\eps) \\
 &\qquad \ge D_H^{\eps - 2 \delta} \bigg(\cN^n\left( \vartheta_{A^{n}A'^{n}}\right) \,\bigg\|\ \sum_{t \in {\cal P}_n^{q,\mu}} p(t)\left( \cN^n(\pi_{A^{ n}}^t)\right) \otimes  \gamma_{n,\mu}  \,\rho_{A'}^{\otimes n}\bigg) - f(\eps, \delta) \\
&\qquad = D_H^{\eps - 2 \delta}\bigg(\cN^n\left( \vartheta_{A^{ n}A'^n}\right) \,\bigg\|\, \sum_{t \in {\cal P}_n^{q,\mu}} p(t)\left( \cN^n(\pi_{A^{ n}}^t)\right) \otimes  \rho_{A'}^{\otimes n}\bigg) - f(\eps, \delta) - \log  \gamma_{n,\mu}\\
&\qquad \ge D_H^{\eps - 2 \delta- g(n,\mu)}\bigg(
\big( \cN(\psi_{AA'}) \big)^{\otimes n} \,\bigg\|\, \sum_{t \in {\cal P}_n^{q,\mu}} p(t)\, \cN^n(\pi_{A^{ n}}^t) \otimes  \rho_{A'}^{\otimes n}\bigg) - f(\eps, \delta) - \log  \gamma_{n,\mu},\\
& \qquad \ge D_H^{\eps - 2 \delta- g(n,\mu)}\bigg(
\big( \cN(\psi_{AA'}) \big)^{\otimes n} \,\bigg\|\, \sum_{t \in {\cal P}_n} p(t)\, \cN^n(\pi_{A^{ n}}^t) \otimes  \rho_{A'}^{\otimes n}\bigg) - f(\eps, \delta) - \log  \gamma_{n,\mu},\\
& \qquad = D_H^{\eps - 2 \delta - g(n,\mu)}\left(
\big( \cN(\psi_{AA'}) \big)^{\otimes n} \,\Big\|\, \Big(\cN(\rho_{A})\right)^{\otimes n} \otimes  \rho_{A'}^{\otimes n}\Big) - f(\eps, \delta) - \log  \gamma_{n,\mu}.
\label{long}
\end{align}
The first and second lines follow from items $1$ and $2$ of Lemma~\ref{props-hypo}, respectively. The third line follows from item $4$ of Lemma~\ref{props-hypo}. The fourth line also follows from item $2$ of Lemma~\ref{props-hypo}, since
\be 
\sum_{t \in {\cal P}_n} p(t) \cN^n(\pi_{A^{ n}}^t) \otimes  \rho_{A'}^{\otimes n} \ge \sum_{t \in {\cal P}_n^{q,\mu}} p(t) \cN^n(\pi_{A^{ n}}^t) \otimes  \rho_{A'}^{\otimes n}.
\ee
The last line follows from the linearity of $ \cN^n$ and the fact that
\be
 \sum_{t \in {\cal P}_n} p(t)\pi_{A^{ n}}^t = \tr_{A'^{ n}} (\psi_{AA'}^{\otimes n}) = \rho_{A}^{\otimes n}.
\ee

Let us choose $\delta =1/{\sqrt{n}}$ and $\mu = \left((|{\cal X}| + 1) \log (n+1)\right)/n$.
Then
\be
 g(n,\mu) = \frac{1}{\sqrt{(n+1)}} \le \frac{1}{\sqrt{n}}, \quad \textrm{and} \quad
\eps - 2 \delta- g(n,\mu) \ge \eps - {3}/{\sqrt{n}}.
\ee
Since $D_H^\eps(\rho\|\sigma) \ge D_H^{\eps'}(\rho\|\sigma)$ for $\eps > \eps'$, we obtain the following bound from \reff{long} 
\be
 \log M_{\rm ea}^*(\cN^n,\eps) \ge D_H^{ \eps - {3}/{\sqrt{n}}}\Big(
\big( \cN(\psi_{AA'}) \big)^{\otimes n} \Big\| \left(\cN(\rho_{A})\right)^{\otimes n} \otimes  \rho_{A'}^{\otimes n}\Big) - f(\eps, \delta) - \log  \gamma_{n,\mu},
\label{one}
\ee
thus arriving at an expression involving the hypothesis testing relative entropy for product states.
Substituting the above choices for $\delta$ and $\mu$ in the expressions \reff{fed} for $f(\eps, \delta)$ and in $\gamma_{n,\mu}$, we find that
\be\label{two}
f(\eps, \delta) + \log  \gamma_{n,\mu} = O(\log n).
\ee

% Finally, we thus arrive at an expression where $D_H$ is evaluated for product states.
Crucially, Lemma~\ref{second-order} applied to \eqref{one} now implies that
\begin{align}
\log M_{\rm ea}^*(\cN^n,\eps) & \ge n D\big(
\cN(\psi_{AA'}) \, \big\|\, \cN(\rho_{A}) \otimes  
\rho_{A'}\big) \nonumber \\
& \qquad  + \sqrt{nV\big(\cN(\psi_{AA'}) \,\big\|\, \cN(\rho_{A}) \otimes  \rho_{A'} \big)}\, \Phi^{-1}(\eps - {3}/{\sqrt{n}}) + K'(n; \cN,\eps,\psi_{AA'}) \\
 &= n I(A':B)_{\omega} + \sqrt{n V(A':B)_{\omega}}\,\Phi^{-1}(\eps - {3}/{\sqrt{n}}) + K'(n; \cN,\eps,\psi_{AA'}) , 
\end{align}
where $K'(n; \cN,\eps,\psi_{AA'}) = O(\log n)$ due to \eqref{two}.
To conclude the proof, note that $\Phi^{-1}$ is continuously differentiable around $\eps > 0$, and thus $\Phi^{-1}(\eps - {3}/{\sqrt{n}}) = \Phi^{-1}(\eps) + O(1/\sqrt{n})$.
\end{proof}

\subsection{Second-Order Converse for Covariant Quantum Channels}
\label{sec_converse}

In this section, we observe that the Gaussian approximation is valid for the
entanglement-assisted capacity of covariant quantum channels (i.e.,
Conjecture~\ref{conj:main} is true for this class of channels). Holevo first
defined the class of covariant quantum channels \cite{H02}, and it is now
known that many channels fall within this class, including depolarizing
channels, transpose depolarizing channels \cite{WH02,FHMV04}, Pauli channels,
cloning channels \cite{B11}, etc. Note that the following argument up to~\eqref{eq:matthews} has already
essentially been proven in Section~III-E\ of Matthews and Wehner \cite{MW12}.
However, we give a brief exposition in this section for completeness. We leave
open the question of determining whether the Gaussian approximation is valid
for the entanglement-assisted capacity of general discrete memoryless quantum channels.

Let $\mathcal{N}_{A\rightarrow B}$ be a quantum channel mapping density
operators acting on an input Hilbert space~$\mathcal{H}_{A}$ to those acting
on an output Hilbert space~$\mathcal{H}_{B}$. Let $G$ be a compact group, and
for every $g\in G$, let $g\rightarrow U_{A}(  g)  $ and
$g\rightarrow V_{B}(  g)  $ be continuous projective
unitary representations of $G$ in $\mathcal{H}_{A}$ and $\mathcal{H}_{B}$,
respectively. Then the channel $\mathcal{N}_{A\rightarrow B}$ is said to be
covariant with respect to these representations if the following relation
holds for all $g\in G$ and input density operators $\rho$:%
\begin{equation}
\mathcal{N}_{A\rightarrow B}\!\left(  U_{A}(  g)  \rho U_{A}^{\dag
}(  g)  \right)  =V_{B}(  g)  \mathcal{N}_{A\rightarrow
B}(  \rho)  V_{B}^{\dag}(  g)  .
\end{equation}
We restrict our attention in this section to covariant channels for which the
representation acting on the input space is irreducible.

In \cite[Thm.~14]{MW12}, Matthews and Wehner establish the following upper
bound on the one-shot entanglement-assisted capacity of a channel
$\mathcal{N}\equiv\mathcal{N}_{A\rightarrow B}$.%
\begin{equation}
\log M_{\text{ea}}^{\ast}(  \mathcal{N},\varepsilon)  \leq
\max_{\rho_{A}}\min_{\sigma_{B}}D_{H}^{\varepsilon}\big(\mathcal{N}(
\phi_{AA^{\prime}}^{\rho})  \big\|\,\rho_{A^{\prime}}\otimes\sigma
_{B}\big),\label{eq:one-shot-converse}%
\end{equation}
where $\phi_{AA^{\prime}}^{\rho}$ is a purification of $\rho_{A}$ and
$\rho_{A^{\prime}}$ is the reduction of $\phi_{AA^{\prime}}^{\rho}$ to
$A^{\prime}$. They also prove~\cite[Thm.~19]{MW12} that the quantity
$\beta_{\varepsilon}\left(  \mathcal{N}_{A\to B}(  \phi
_{AA^{\prime}}^{\rho})  \|\rho_{A^{\prime}}\otimes\sigma_{B}\right)
$\ defined through \eqref{beta} is convex in the input density operator $\rho_{A}$ for any $\sigma_{B}$,
from which it follows that the quantity%
\begin{equation}
\beta_{\varepsilon}(  \mathcal{N}_{A\to B},\rho_{A})
:=\max_{\sigma_{B}}\beta_{\varepsilon}\big(  \mathcal{N}_{A\to
B}(  \phi_{AA^{\prime}}^{\rho}) \big\|\,\rho_{A^{\prime}}\otimes
\sigma_{B}\big)
\end{equation}
is convex in $\rho_{A}$ because it is the pointwise maximum of a set of convex functions.

We would now like to apply these results to the entanglement-assisted capacity
of any discrete memoryless covariant channel $\mathcal{N}_{A^{n}\rightarrow
B^{n}}\equiv\mathcal{N}^{\otimes n}$.
%All such channels are covariant with
%respect to permutations of the input systems $A^n$, so that%
%\begin{equation}
%\mathcal{N}_{A^{n}\rightarrow B^{n}}\left(  W_{A^{n}}^{\pi}\rho_{A^{n}}\left(
%W_{A^{n}}^{\pi}\right)  ^{\dag}\right)  =W_{B^{n}}^{\pi}\mathcal{N}%
%_{A^{n}\rightarrow B^{n}}\left(  \rho_{A^{n}}\right)  \left(  W_{B^{n}}^{\pi
%}\right)  ^{\dag},
%\end{equation}
%where $W_{A^{n}}^{\pi}$ is a unitary representation of a permutation
%$\pi$ in the symmetric group $S_{n}$, which effects
%a permutation of the $A^{n}$ systems (with a similar convention for $W_{B^{n}}^{\pi}$).
By definition, such channels have the following covariance:%
\begin{multline}
\mathcal{N}_{A^{n}\rightarrow B^{n}}\left(  \left[  U_{A_{1}}(
g_{1})  \otimes\cdots\otimes U_{A_{n}}(  g_{n})  \right]
\rho_{A^{n}}\left[  U_{A_{1}}(  g_{1})  \otimes\cdots\otimes
U_{A_{n}}(  g_{n})  \right]  ^{\dag}\right)  \\
=\left[  V_{B_{1}}(  g_{1})  \otimes\cdots\otimes V_{B_{n}}(
g_{n})  \right]  \mathcal{N}_{A^{n}\rightarrow B^{n}}(  \rho
_{A^{n}})  \left[  V_{B_{1}}(  g_{1})  \otimes\cdots\otimes
V_{B_{n}}(  g_{n})  \right]  ^{\dag}.
\end{multline}
Let $T_{A^{n}}$ be a shorthand for
%any unitary which is the composition of a permutation unitary
%$W_{A^{n}}^{\pi}$ and
a sequence of local unitaries of the form $U_{A_{1}}(  g_{1})
\otimes\cdots\otimes U_{A_{n}}(  g_{n})  $. Let $\mathbb{E}$ denote
the expectation over all such unitaries $T_{A^{n}}$, with
the measure being the product Haar measure $\mu(g_1) \times \cdots \times
\mu(g_n)$. Then
following~\cite[Sec.~III-E]{MW12}, we can conclude the following chain of
inequalities:%
\begin{align}
\beta_{\varepsilon}(  \mathcal{N}_{A^{n}\rightarrow B^{n}},\rho_{A^{n}%
})   &  =\mathbb{E}\left\{  \beta_{\varepsilon}\left(  \mathcal{N}%
_{A^{n}\rightarrow B^{n}},T_{A^{n}}\rho_{A^{n}}T_{A^{n}}^{\dag}\right)
\right\}  \\
&  \geq\beta_{\varepsilon}\left(  \mathcal{N}_{A^{n}\rightarrow B^{n}%
},\mathbb{E}\left\{  T_{A^{n}}\rho_{A^{n}}T_{A^{n}}^{\dag}\right\}  \right)
\\
&  =\beta_{\varepsilon}\left(  \mathcal{N}_{A^{n}\rightarrow B^{n}},\pi
_{A_{1}}\otimes\cdots\otimes\pi_{A_{n}}\right)  ,
\end{align}
where $\pi$ is the maximally mixed state.
The first equality is a result of \cite[Prop.~29]{MW12} (this follows
directly from the assumption of channel covariance with respect to the
operations $T_{A^{n}}$). The sole inequality exploits convexity as mentioned
above. The last equality follows because the state
$\mathbb{E}\big\{  T_{A^{n}}\rho_{A^{n}}T_{A^{n}}^{\dag}\big\}$ commutes with
all local unitaries $U_{A_{1}%
}(  g_{1})  \otimes\cdots\otimes U_{A_{n}}(  g_{n})  $.
As a consequence of Schur's lemma and the irreducibility of the representation on the input space, the only state which possesses such invariances is the
tensor-power maximally mixed state.
Note that we require irreducibility of the representation on only the input
space in order for this argument to hold.
So, by using the definition of $D_{H}^{\varepsilon}$, we can then conclude
that%
\begin{align}
\log M_{\text{ea}}^{\ast}(  \mathcal{N}_{A^{n}\rightarrow B^{n}%
},\varepsilon)   &  \leq\max_{\rho_{A^{n}}}\min_{\sigma_{B^{n}}}%
D_{H}^{\varepsilon}(  \mathcal{N}_{A^{n}\rightarrow B^{n}}(
\phi_{A^{n}A^{\prime n}}^{\rho})  \big\|\, \rho_{A^{\prime n}}\otimes
\sigma_{B^{n}})  \\
&  \leq\min_{\sigma_{B^{n}}}D_{H}^{\varepsilon}(  (  \mathcal{N}%
_{A\rightarrow B}(  \Phi_{AA^{\prime}})  )^{\otimes n}%
\big\|\, \pi_{A^{\prime}}^{\otimes n}\otimes\sigma_{B^{n}})  \\
&  \leq D_{H}^{\varepsilon}(  (  \mathcal{N}_{A\rightarrow B}(
\Phi_{AA^{\prime}})  )^{\otimes n} \big\|\, \pi_{A}^{\otimes n}%
\otimes\left[  \mathcal{N}_{A\rightarrow B}\left(  \pi_{A}\right)  \right]
^{\otimes n})  \label{eq:matthews}\\
&  =n I(  A^{\prime}:B)_{\omega}  +\sqrt{n V(  A^{\prime}:B)_{\omega}  }\,%
\Phi^{-1}(  \varepsilon)  +O(  \log n) ,
\end{align}
where the information quantities in the final line are with respect to the
state $\omega_{A'B} := \mathcal{N}_{A\rightarrow B}(  \Phi_{AA^{\prime}})  $. The final equality uses the asymptotic expansion in Lemma~\ref{second-order}.

\section{Discussion}
\label{sec_discussion}

We have established the direct part of the Gaussian approximation in Theorem~\ref{th:main} and conjectured that the converse also holds in Conjecture~\ref{conj:main}.
We again note that all of our results apply to entanglement-assisted quantum communication as well, due to the teleportation \cite{BBCJPW93} and super-dense coding \cite{PhysRevLett.69.2881} protocols and the results of \cite{LM14}.
In the following we will discuss some of the approaches taken and difficulties encountered when trying to prove the converse for general channels.

\begin{description}
  \item[Arimoto Converse:] Converse proofs using Arimoto's approach~\cite{arimoto73} and quantum generalizations of the R\'enyi divergence~\cite{lennert13,wilde13} as in~\cite{GW13} can be used to establish that the probability of successful decoding goes to zero exponentially fast for codes with $\frac{1}{n} \log |M| > C_{\textrm{ea}}$. However, they only yield trivial results when $\frac{1}{n} \log |M| = C_{\textrm{ea}} \pm O(1/\sqrt{n})$, as is the case in the Gaussian approximation.
  \item[De Finetti Theorems:] Following Matthews and Wehner~\cite{MW12}, we find the following converse bound for $n$ uses of the channel employing the arguments presented in Section~\ref{sec_converse} and~\eqref{eq:one-shot-converse}.
  \begin{align}
    \log M_{\text{ea}}^*(\cN^{n}, \eps) \leq \max_{\rho_{A^n}}\min_{\sigma_{B^n}}%
D_{H}^{\varepsilon}\big(  \cN (
\phi_{A^nA'^n}^{\rho})  \big\|\, \rho_{A'^n}\otimes\sigma_{B^n}\big) ,
  \end{align}
  where $\rho_{A^n}$ and $\sigma_B^n$ are invariant under permutations of the $n$ systems, and $\phi_{A^nA'^n}^{\rho}$ is chosen to have this property as well. One may now try to approximate the state $\phi_{A^nA'^n}$ by a convex combination of product states using the de Finetti theorem or the exponential de Finetti theorem~\cite{renner07}. However, the problem is that the number of systems that need to be sacrificed is at least of the order $\sqrt{n}$, and thus affects the second-order term significantly. 
  %Finally, it is not clear how the post-selection technique~\cite{christandl09} can be applied to this situation. 
  \item[Relation to Channel Simulation:] EAC coding is closely related to the classical communication cost in entanglement-assisted channel simulation~\cite{BDHSW12} and~\cite{BCR09}. In the latter paper, some bounds on the classical communication cost of entanglement-assisted channel simulation for a finite number of channels $n$ are given. However, these bounds turn out to be unsuitable for our purposes since the error is scaled by a factor polynomial in $n$ as a result of applying the
  post-selection technique~\cite{christandl09}. It is not clear how the proof in~\cite{BCR09} can be adapted to yield a statement for fixed error.
\end{description}

We believe that establishing Conjecture~\ref{conj:main} thus requires new techniques and that this constitutes an interesting open problem. 

\paragraph*{Acknowledgements.} We are especially grateful to Milan Mosonyi for
insightful discussions and for his help in establishing the proof of
Propositions~\ref{thm-one-shot} and \ref{pr:direct-second}.
We acknowledge discussions with Mario Berta, Ke~Li, Will Matthews, and Andreas Winter, \MT{and we thank
the Isaac Newton Institute
(Cambridge) for its hospitality while part of this work was completed.}
MT is funded by the Ministry of Education (MOE) and National Research Foundation Singapore, as well as MOE Tier 3 Grant ``Random numbers from quantum processes'' (MOE2012-T3-1-009). MMW acknowledges startup funds from the Department of
Physics and Astronomy at LSU, support from the NSF through Award No.~CCF-1350397, and support from the DARPA Quiness Program through US
Army Research Office award W31P4Q-12-1-0019.

\appendix

\section{Heisenberg-Weyl Operators}
\label{app:hw}

For any $x,z \in \{0,1,\cdots, d\}$ the Heisenberg-Weyl Operators $X(x)$ and $Z(z)$ are defined through their actions on the vectors
of the qudit computational basis $\{|j\rangle \}_{j \in \{0,1,\cdots, d-1\}}$ as follows:
%(see, e.g.,~\cite{MW13}):
\begin{align}\label{Weyl} 
X(x) |j\rangle &= |j \oplus x\rangle,\\
Z(z)|j\rangle &= e^{2\pi i z j / d} |j\rangle ,
\end{align}
where $j \oplus x = (j + x) \, {\rm mod }\, d $. Also note that if $d = 1$, then both $X(x)$ and $Z(z)$ are equal to the identity operator.

\section{The Method of Types}
\label{sc:types}

In our proofs we employ the notion of {\em{types}}~\cite{csiszar98}, and hence we briefly recall certain relevant definitions and properties here. 

Let ${\cal X}$ denote a discrete alphabet and fix $n \in \mathbb{N}$. The {\em{type}} (or empirical probability distribution)  $P_{x^n}$ of a sequence $x^n \in {\cal X}^n$ is the empirical frequency of occurrences of each letter of ${\cal X}$, i.e., $P_{x^n}(a) := \frac1{n} \sum_{i=1}^n \delta_{x_i,a}$
%$$P_{x^n}(a) = N(a|x^n)/n$$ 
for all $a \in {\cal X}$.
%, where $N(a|x^n)$ is the number of times the letter $a$ occurs in the sequence $x^n \in {\cal X}^n$. 
Let ${\cal P}_n$ denote the set of all types. The number of types, $|{\cal P}_n|$, satisfies the bound~\cite[Thm.~11.1.1]{CoverThomas}
\be\label{type1}
|{\cal P}_n| \le (n+1)^{|{\cal X}|}.
\ee
For any type $t \in {\cal P}_n$, the {\em{type class}} $\cT^t$ of $t$ is the set of sequences of type $t$, i.e.\ 
\be
\cT^t := \{x^n \in {\cal X}^n \, : \, P_{x^n} = t \}.
\ee
The number of types in a type class $\cT^t$ satisfies the following lower bound~\cite[Lm.~II.2]{csiszar98}:
\be\label{type2}
|\cT^t| \ge \frac{2^{nH(t)}}{(n+1)^{|{\cal X}|}},
\ee
where $H(t) := - \sum_{a \in {\cal X}} t(a) \log t(a)$, is the Shannon entropy of the type.

Let $q$ be any probability distribution on ${\cal X}$. For any sequence $x^n = (x_1, x_2, \ldots, x_n) \in {\cal X}^n$,
%of type $t$ (i.e., with $P_{x^n}=t$), 
let $q^n(x^n) = \prod_{i=1}^n q(x_i)$. Then, we have
\be\label{type3}
q^n(x^n)  = 2^{-n\left(H(t) + D(t\|q)\right)}, \qquad \textrm{where} \quad t = P_{x^n} 
\ee
is the type of $x^n$ and $D(t\|q):= \sum_{a \in {\cal X}} t(a) \log \frac{t(a)}{q(a)}$ is the Kullback-Leibler divergence of the probability distributions $t$ and $q$.
From \reff{type1}, \reff{type2} and \reff{type3} it follows that for any sequence $x^n \in {\cal X}^n$ of type $t$, 
\be \label{type4}
(n+1)^{|{\cal X}|} 2^{nD(t\|q)} q^n(x^n) = 2^{-nH(t)}(n+1)^{|{\cal X}|}\ge \frac{1}{|\cT^t|}.
\ee

Finally, for any $\mu > 0$ we have~\cite[Eq.~(11.98)]{CoverThomas}
\be\label{type6}
  \sum_{x^n \in \cX^n \atop D(P_{x^n} \| q) > \mu} q^n(x^n) \le 2^{-n\left( \mu - |{\cal X}|\frac{ \log (n+1)}{n}\right)}.
\ee

\bibliography{Ref}
\bibliographystyle{alpha}

\end{document}